\documentclass[12pt]{amsart}

\usepackage{amsfonts}

\newtheorem{Theorem}{Theorem}

\newtheorem{Proposition}{Proposition}
\newtheorem{Lemma}{Lemma}

\newcommand{\LL}{{\mathrm{L}}}
\newcommand{\Id}{{\mathbf{1}}}
\newcommand{\dom}{{\mathrm{dom}~}}
\newcommand{\N}{\ensuremath{{\mathbb N}}}
\newcommand{\Z}{\ensuremath{{\mathbb Z}}}
\newcommand{\R}{\ensuremath{{\mathbb R}}}
\newcommand{\C}{\ensuremath{{\mathbb C}}}

\newcommand{\cte}{{\mathrm{cte}~}}

\newcommand{\hil}{{\mathcal H}~}

\begin{document}

\title[Quantum Energy Expectation]{Quantum Energy Expectation in Periodic
Time-Dependent hamiltonians via Green Functions}
\author{C\'{e}sar R. de Oliveira}
\address{Departamento de Matem\'{a}tica -- UFSCar, S\~{a}o Carlos, SP,
13560-970 Brazil, and 
Department of Mathematics, University of British Columbia,
Vancouver, BC, V6T 1Z2 Canada} \email{oliveira@dm.ufscar.br, oliveira@math.ubc.ca}
\author{Mariza S. Simsen}
\address{Instituto de Ci\^encias Exatas, Universidade Federal de
Itajub\'a,    Itajub\'a,  MG, 37500-000 Brazil \\}
\email{mariza@unifei.edu.br}
\thanks{The authors thank the anonymous referees for careful readings and suggestions. CRdeO
thanks the warm hospitality of PIMS (Vancouver) and the support by
CAPES (Brazil).} 

\begin{abstract} Let $U_F$ be the Floquet operator of a time periodic
hamiltonian $H(t)$.
 For each positive  and discrete observable $A$ (which we call a {\em
probe energy}), we derive  a formula for the Laplace time average
of its expectation value up to time $T$ in terms of its
eigenvalues and Green functions at the circle of radius $e^{1/T}$.
Some simple applications are provided which support its
usefulness.
\end{abstract}
\maketitle

\tableofcontents

\section{Introduction} Consider a general periodically driven quantum
hamiltonian system
\[
H(t)=H_0+V(t)
\] with period $\tau$ acting in a
separable Hilbert space $\mathcal{H}$ and let $U_F$ denote its
Floquet operator, so that if  $\xi$ is the initial state (at time
zero) of the system then $U_F^m\xi$ is this  state at time
$m\tau$. Typically, the unperturbed hamiltonian $H_0$ is assumed
to have purely point spectrum so that the same is true for
$e^{-i\tau H_0}$. What happens when $H_0$ is perturbed by
$V(t)$? A natural physical question is if the expectation values of
the unperturbed energy $H_0$ remain bounded when $V(t)\neq0$. This
question is formulated based on many physical models, in
particular on the Fermi accelerator  in which a  particle can
acquire unbounded energy from collisions with a heavy periodically
moving wall. Here quantum mechanics is considered and, more
precisely, if
\[\sup_{m\in\N}\left|\left\langle
U_F^m\xi,H_0U_F^m\xi\right\rangle\right|
\] is finite or not, where $\xi\in\dom H_0\subset\mathcal{H}$, the domain of
$H_0$.

Motivated by models with  hamiltonians as above
 $H(t)=H_0+V(t)$, one is suggested to probe quantum (in)stability through
the behavior of an ``abstract energy operator" which we call a
{\em probe operator} and will be represented by a positive,
unbounded, self-adjoint operator $A:\dom
A\subset\mathcal{H}\rightarrow\mathcal{H}$ and with discrete
spectrum, \[ A\varphi_n=\lambda_n\varphi_n,
\]
$0\leq\lambda_n<\lambda_{n+1}$, such that if $U_F^m\xi\in\dom A$
for all $m\in\N$, then, for each~$m$, the expectation value
$E_{\xi}^A(m)=\langle U_F^m\xi,AU_F^m\xi\rangle$ is finite. It is
convenient to write $E_\xi^A(m)=+\infty$ if $U_F^m\xi$ does not
belong to $\dom A$.

We say the system is {\em $A$-dynamically stable} when
$E_{\xi}^A(m)$ is a bounded function of time $m$, and {\em
$A-$dynamically unstable} otherwise (usually we say just {\em
(un)stable}). If the function $E_{\xi}^A(m)$ is not bounded one can ask about
its asymptotic behavior, that is, how does $E_{\xi}^A(m)$ behave
as $m$ goes to infinity? Usually this is a very difficult question and
sometimes the temporal average of $E_{\xi}^A(m)$ is considered, as
we will do in this work.

Quantum systems governed by a time periodic hamiltonian have their
dynamical stability often characterized in terms of the spectral
properties of the corresponding Floquet operator. As in the
autonomous case, the presence of continuous spectrum is a
signature of unstable quantum systems; this is a rigorous consequence of
the famous RAGE theorem~\cite{ISTQD}, firstly proved for the autonomous
case and then for time-periodic
hamiltonians~\cite{EV,YK}. At first sight a Floquet operator with
purely point spectrum would imply stability, but one should be
alerted by examples with purely point spectrum and dynamically
unstable \cite{dRJLS,JSS,dOP} in the autonomous case and, recently,
also a time-periodic model with energy instability~\cite{deOS} was found.

Dynamical stability of time-dependent systems was studied, e.g.,
in references \cite{EV,BJLPN,dO,N,J,dBF,BJ,ADE,dOT,G,L,G1,BCM}.
In~\cite{ADE} it was proved that the applicability of the KAM
method gives a uniform bound at the expectation value of the
energy for a class of time-periodic hamiltonians considered
in~\cite{DS}.

For hamiltonians $H(t)=H_0+V(t)$, not necessarily periodic, with
$H_0$ a positive self-adjoint operator whose spectrum consists of
separated bands $\{\sigma_j\}_{j=1}^{\infty}$ such that
$\sigma_j\subset[\lambda_j,\Lambda_j]$, upper bounds of the type
\[\langle U(t,0)\psi,H_0U(t,0)\psi\rangle\leq \mathrm{cte}\;
t^{\frac{1+\alpha}{n\alpha}}\] were obtained in~\cite{N} if the
gaps $\lambda_{j+1}-\Lambda_j$ grow like $j^{\alpha}$, with
$\alpha>0$, and if $V(t)$ is strongly $C^n$ with
$n\geq[\frac{1+\alpha}{2\alpha}]+1$. The proof is based on
adiabatic techniques that require smooth time dependence and
therefore do not apply to kicked systems. In \cite{J,BJ} upper
bounds complementary to those of~\cite{N} described above are
obtained.

In~\cite{EV,BJLPN,dO,dOT} stability results are obtained through
topological properties of the orbits $\xi(t)=U(t,0)\xi$ for
$\xi\in\mathcal{H}$, while in~\cite{G,L,G1,BCM} lower bounds for averages of
the type
\[\frac{1}{T}\sum_{m=1}^{T}\left\langle U_F^m\xi,H_0U_F^m\xi\right\rangle\geq
CT^{\gamma}\] are obtained for periodic hamiltonians
$H(t)=H_0+V(t)$ through dimensional properties of the spectral
measure $\mu_{\xi}$ associated with $U_F$ and $\xi$ (the exponent
$\gamma$ depends on the measure $\mu_{\xi}$).

In this work we study (in)stability of periodic time-dependent
systems. As for tight-binding models (see~\cite{DST} and
references therein) we consider the Laplace-like average of
$\langle U_F^m\xi,AU_F^m\xi\rangle$, that is,
\[\frac{2}{T}\sum_{m=0}^{\infty}e^{-\frac{2m}{T}}\langle
U_F^m\xi,AU_F^m\xi\rangle,\] where $A$ is a probe energy,
$\xi$ is an element of $\dom A$ and $U_F$ is the Floquet operator.
The main technical reason for working with this expression for the
time average is that it can be written in terms of (see
Theorem~\ref{FormulaTheorem}) the eigenvalues of $A$, i.e.,
$A\varphi_j=\lambda_j\varphi_j$, and the matrix elements $\langle
\varphi_j,R_z(U_F)\xi\rangle$ of the resolvent operator
$R_z(U_F)=(U_F-z\Id)^{-1}$ (with $z=e^{-iE}e^{1/T}$) with respect
to the orthonormal basis $\{\varphi_j\}$ of the Hilbert space
(here $\Id$ denotes the identity operator). Lemma~\ref{lema.49}
relates the long run of Laplace-like average with the usual
Ces\`{a}ro average. In Section~\ref{FormulaSection} we shall prove
this abstract results and present some applications in
Section~\ref{ApplicationsSection}.

Since our main results are for temporal Laplace averages of
expectation values of probe energies (see Section~\ref{sectAEGF}),
in practice we will think of (in)stability in terms of
(un)boundedness of such averages. Note that
unbounded Laplace averages  imply unboundedness of expectation values of probe energies themselves.

\

\section{Average Energy and Green Functions}\label{sectAEGF}
\label{FormulaSection} Consider a time-dependent hamiltonian
$H(t)$ with $H(t+\tau)=H(t)$ for all $t\in\R$, acting in the
separable Hilbert space $\mathcal H$. Suppose the existence of the
propagators $U(t,s)$, so that  the Floquet operator
$U_F=U(\tau,0)$ is at our disposal. Let $A$ be a probe energy and
$\lambda_j,\varphi_j$ as in the Introduction. Also,
$\{\varphi_j\}_{j=1}^{\infty}$ is an orthonormal basis of
$\mathcal{H}$.

The main interest is in the study of the expectation values, herein defined by
\[
E_{\xi}^A(m) := \left\{ {\begin{array}{*{20}c}
  \langle U_F^m\xi,AU_F^m\xi\rangle, \quad\mathrm{if}\;\;\;\; U_F^m\xi\in\dom
A,  \\
 \quad\;\;\; +\infty, \qquad \mathrm{if}\;\;\;\;U_F^m\xi\in{\mathcal
H}\setminus\dom A, \\
\end{array}} \right.
\]
as function of time $m\in\N$. Another quantity of interest is the
time dependence of the moment of this probe energy, which takes
values in $[0,+\infty]$ and is defined by
\[M_{\xi}^A(m) :=
\sum_{j=1}^{\infty}\lambda_j\left|\langle\varphi_j,U_F^m\xi\rangle\right|^2.\]
Our first remark is the equivalence of both concepts (under certain circumstances).

\

\begin{Proposition}\label{EqualityEquation} If $U_F^m\xi\in\dom A$ for all $m$, then
\[E_{\xi}^A(m)=M_{\xi}^A(m),\qquad\forall m.\]
This holds, in particular, if $\dom A$ is invariant under the time evolution $U_F^m$ and $\xi\in\dom A$.
\end{Proposition}
\begin{proof} Since $\dom A\subset \dom A^{\frac12}$ \cite{ISTQD} one has
$U_F^m\xi\in\dom A^{\frac{1}{2}} $, for all $m$, and~so
\begin{eqnarray*}
M_{\xi}^A(m) &=&
\sum_{j=1}^{\infty}|\langle\lambda_j^{\frac{1}{2}}\varphi_j,U_F^m\xi\rangle|^2=
\sum_{j=1}^{\infty}|\langle
A^{\frac{1}{2}}\varphi_j,U_F^m\xi\rangle|^2\\ &=&
\sum_{j=1}^{\infty}|\langle\varphi_j,A^{\frac{1}{2}}U_F^m\xi\rangle|^2
= \|A^{\frac{1}{2}}U_F^m\xi\|^2\\ &=& \langle
A^{\frac{1}{2}}U_F^m\xi,A^{\frac{1}{2}}U_F^m\xi\rangle=\langle
U_F^m\xi,AU_F^m\xi\rangle= E_{\xi}^A(m),
\end{eqnarray*}
which is the stated result.
\end{proof}

We introduce the temporal Laplace average of $E_{\xi}^A$ (see also
the Appendix) by the following function of $T>0$, which also takes
values in $[0,+\infty]$,
\begin{equation}\label{LaplaceAverage}
L_{\xi}^A(T) :=
\frac{2}{T}\sum_{m=0}^{\infty}e^{-\frac{2m}{T}}E_{\xi}^A(m).
\end{equation}
 Under certain conditions, the next result shows that the upper $\beta^+$
and lower $\beta^-$ growth exponents of this average, that is,
roughly they are the best exponents so that for large~$T$ there
exist $0\le c_1\le c_2<\infty$ with
\[
c_1\,T^{\beta^-}\le L_\xi^A(T) \le c_2\,T^{\beta^+},
\] and the corresponding exponents  for the temporal Ces\`{a}ro average
\[
C_{\xi}^A(T)=\frac{1}{T}\sum_{m=0}^TE_{\xi}^A(m)
\] are closely related; this
follows at once by Lemma~\ref{lema.49}, which perhaps could be
improved to get equality also between lower exponents. Note that, although not indicated, these exponents depend on the initial condition~$\xi$.

\

\begin{Lemma}\label{lema.49} If $(h(m))_{m=0}^{\infty}$ is a nonnegative
sequence, and $h(m)\leq Cm^n$ for some $C>0$ and $n\geq0$, then
$\beta^{+}_e=\beta^{+}_d$ and $\beta^{-}_e\le \beta^{-}_d$, where
\[\beta^{+}_e=\limsup_{T\rightarrow\infty}\frac{\log(\sum_{m=0}^Th(m))}{\log
T},\qquad
\beta^{-}_e=\liminf_{T\rightarrow\infty}\frac{\log(\sum_{m=0}^Th(m))}{\log
T},\]
\[\beta^{+}_d=\limsup_{T\rightarrow\infty}\frac{\log(\sum_{m=0}^{\infty}
e^{-\frac{2m}{T}}h(m))}{\log T},\;\;\;
\beta^{-}_d=\liminf_{T\rightarrow\infty}\frac{\log(\sum_{m=0}^{\infty}
e^{-\frac{2m}{T}}h(m))}{\log T}.\]
\end{Lemma}
\begin{proof}
Note that for $0\leq m\leq T$ we have $e^{-2}\leq
e^{-\frac{2m}{T}}\leq1$, and so
\[\sum_{m=0}^Th(m)\leq\sum_{m=0}^Te^2e^{-\frac{2m}{T}}h(m)\leq
e^2\sum_{m=0}^{\infty}e^{-\frac{2m}{T}}h(m).\] Hence
$\beta^{\pm}_e\leq\beta^{\pm}_d$.

On the other hand, for each $\epsilon>0$, denoting by $\lceil
x\rceil$ the smallest integer larger or equal to $x$, one has
\begin{eqnarray*}
\sum_{m=0}^{\infty}e^{-\frac{2m}{T}}h(m) &=& \sum_{m=0}^{\lceil
T^{1+\epsilon}\rceil}e^{-\frac{2m}{T}}h(m)+\sum_{m=\lceil
T^{1+\epsilon}\rceil+1}^{\infty}e^{-\frac{2m}{T}}h(m)\\ &\leq&
\sum_{m=0}^{\lceil T^{1+\epsilon}\rceil}h(m)+C\sum_{m=\lceil
T^{1+\epsilon}\rceil+1}^{\infty}e^{-\frac{2m}{T}}m^n.
\end{eqnarray*}
Now, for $T$ large enough $\frac{nT}{2}<T^{1+\epsilon}\leq\lceil
T^{1+\epsilon}\rceil$. Thus
\[\sum_{m=\lceil
T^{1+\epsilon}\rceil+1}^{\infty}e^{-\frac{2m}{T}}m^n\leq\int_{\lceil
T^{1+\epsilon}\rceil}^{\infty}e^{-\frac{2t}{T}}t^ndt.\] Therefore,
for each $\epsilon>0$ and $T$ large enough
\begin{eqnarray*}
\sum_{m=0}^{\infty}e^{-\frac{2m}{T}}h(m) &\leq& \sum_{m=0}^{\lceil
T^{1+\epsilon}\rceil}h(m)+C\int_{\lceil
T^{1+\epsilon}\rceil}^{\infty}e^{-\frac{2t}{T}}t^ndt\\ &\leq&
\sum_{m=0}^{\lceil
T^{1+\epsilon}\rceil}h(m)+\tilde{C}e^{-2T^{\epsilon}}T^n.
\end{eqnarray*}
Since $e^{-2T^{\epsilon}}T^n\rightarrow0$ as $T\rightarrow\infty$,
it follows that
\begin{eqnarray*}
\beta^{+}_d &=&
\limsup_{T\rightarrow\infty}\frac{\log\sum_{m=0}^{\infty}
e^{-\frac{2m}{T}}h(m)}{\log T}\\ &\leq&
\limsup_{T\rightarrow\infty}\frac{\log\sum_{m=0}^{\lceil
T^{1+\epsilon}\rceil}h(m)}{\log T}\\
&=& \limsup_{T\rightarrow\infty}\frac{\log\sum_{m=0}^{\lceil
T^{1+\epsilon}\rceil}h(m)}{\log \lceil
T^{1+\epsilon}\rceil}\,\frac{\log\lceil T^{1+\epsilon}\rceil}{\log T}\\
&\leq& \limsup_{T\rightarrow\infty}\frac{\log\sum_{m=0}^{\lceil
T^{1+\epsilon}\rceil}h(m)}{\log\lceil
T^{1+\epsilon}\rceil}\,\frac{\log (T+1)^{1+\epsilon}}{\log T}\\
&=&
(1+\epsilon)\limsup_{T\rightarrow\infty}\frac{\log\sum_{m=0}^{\lceil
T^{1+\epsilon}\rceil}h(m)}{\log\lceil T^{1+\epsilon}\rceil}\\
&\leq& (1+\epsilon)\beta^{+}_{e}.
\end{eqnarray*}
As $\epsilon>0$ was arbitrary, $\beta^{+}_{d}\leq\beta^{+}_{e}$.
\end{proof}

Recall that the Green functions $G_z^{\xi}(j)$ associated with the
operators $A,U_F$ at $\xi\in\mathcal{H}$ and $z\in\C, |z|\ne1$, are
defined by the matrices elements of the resolvent operator
$R_z(U_F)=(U_F-z\Id)^{-1}$ along the orthonormal basis
$\{\varphi_j\}_{j=1}^{\infty}$, that~is,
\[G_z^{\xi}(j) := \langle\varphi_j,R_z(U_F)\xi\rangle.\]
Note that $G_z^{\xi}(j)$ is always well defined since for
$|z|\ne1$ that resolvent operator is bounded.
Theorem~\ref{FormulaTheorem} is the main reason for considering
the temporal averages $L_{\xi}^A(T)$. It presents a formula that
translates the Laplace average of wavepackets at time~$T$ into an
integral of the Green functions over ``energies'' in the circle of radius $e^{1/T}$ in the complex plane (centered at the origin). As $T$ grows the integration region
approaches the unit circle where the spectrum of $U_F$ lives and
$R_z(U_F)$ takes singular values, so that (hopefully)
$A$-(in)stability can be quantitatively detected.

\

\begin{Theorem}\label{FormulaTheorem} Assume that $U_F^m\xi\in\dom A$ for
all $m\ge0$. Then
\begin{equation}\label{formulaDoTeor}
L_{\xi}^A(T)=\frac{1}{\pi e^{-\frac{2}{T}}}
\frac{1}{T}\sum_{j=1}^{\infty}\lambda_j\int_0^{2\pi}\left|G_z^{\xi}(j)\right|^2dE,\qquad
z=e^{-iE+\frac{1}{T}}.
\end{equation}
\end{Theorem}

Before the proof of this theorem, we underline that this formula,
that is, the expression on the right hand side
of~(\ref{formulaDoTeor}), is a sum of positive terms and so it is 
well defined for all $\xi\in\hil$ if we let it take values in $[0,+\infty];$ hence, in
principle it can happen that this formula is finite even for
vectors $U_F^m\xi$ not in the domain of~$A$, where
$L_{\xi}^A(T)=+\infty$. The general case, i.e., $\forall
\xi\in\hil$, can then be gathered in the following inequality
\begin{equation}\label{formulaGeneR}
L_{\xi}^A(T)\ge\frac{1}{\pi e^{-\frac{2}{T}}}
\frac{1}{T}\sum_{j=1}^{\infty}\lambda_j\int_0^{2\pi}\left|G_z^{\xi}(j)\right|^2dE,\qquad
z=e^{-iE+\frac{1}{T}},
\end{equation} so that lower bound estimates for this formula always imply
lower bound estimates for the Laplace average.

\begin{proof}(Theorem~\ref{FormulaTheorem}) First note that, by
hypothesis, $U_F^m\xi\in\dom A^{\frac{1}{2}}$ for each $m\in\N$.
Denote by $\mu_j$ the spectral measure of $U_F$ associated with
the pair $(\varphi_j,\xi)$ and by $\mathcal{F}$ the Fourier
transform $\mathcal{F}:\LL^2[0,2\pi]\rightarrow l^2(\Z)$. By the
spectral theorem for unitary operators
\[\langle\varphi_j,U_F\xi\rangle=\int_0^{2\pi}e^{-iE'}d\mu_j(E').\]
For each $j$ let $a^{(j)}=(a^{(j)}(m))_{m\in\Z}$ be the sequence
\[a^{(j)}(m)=\left\{\begin{array}{ccc} 0 &
\mbox{if} & m<0\\
e^{-\frac{m}{T}}\int_0^{2\pi}e^{-iE'm}d\mu_j(E') & \mbox{if} &
m\geq0 \end{array}\right. .\] Since $a^{(j)}\in l^1(\Z)\cap
l^2(\Z)$ and $\mathcal{F}$ is a unitary operator, it  follows that
$\|a^{(j)}\|_{l^2(\Z)}=\|\mathcal{F}^{-1}a^{(j)}\|_{\LL^2[0,2\pi]}$
and also
\begin{eqnarray*}
(\mathcal{F}^{-1}a^{(j)})(E) &=&
\frac{1}{\sqrt{2\pi}}\sum_{m=-\infty}^{\infty}e^{iEm}a^{(j)}(m)\\
&=&
\frac{1}{\sqrt{2\pi}}\sum_{m=0}^{\infty}e^{iEm}e^{-\frac{m}{T}}\int_0^{2\pi}e^{-iE'm}d\mu_j(E')\\
&=& \frac{1}{\sqrt{2\pi}}\int_0^{2\pi}\Big(\sum_{m=0}^{\infty}
e^{im(E-E')+\frac{i}{T}}\Big)d\mu_j(E')\\ &=&
\frac{1}{\sqrt{2\pi}}\int_0^{2\pi}\frac{1}{1-e^{i(E-E'+\frac{i}{T})}}d\mu_j(E')\\
&=&
\frac{1}{\sqrt{2\pi}}\int_0^{2\pi}\frac{d\mu_j(E')}{e^{i(E+\frac{i}{T})}
(e^{-i(E+\frac{i}{T})}-e^{-iE'})}\\ &=&
-\frac{1}{\sqrt{2\pi}\,e^{i(E+\frac{i}{T})}}\int_0^{2\pi}\frac{d\mu_j(E')}{
e^{-iE'}-e^{-i(E+\frac{i}{T})}}\\ &=&
-\frac{1}{\sqrt{2\pi}\,e^{i(E+\frac{i}{T})}}\langle\varphi_j,
R_{z}(U_F)\xi\rangle\\ &=&
-\frac{1}{\sqrt{2\pi}\,e^{iE}e^{-\frac{1}{T}}}\,G_{z}^{\xi}(j),
\end{eqnarray*}
with $z=e^{-iE+\frac{1}{T}}$. Therefore
\[\left|\mathcal{F}^{-1}a^{(j)}\right|^2(E)=\frac{1}{2\pi
e^{-\frac{2}{T}}}\left|G_{z}^{\xi}(j)\right|^2,
\] and so
\[\left\|\mathcal{F}^{-1}a^{(j)}\right\|^2_{\LL^2[0,2\pi]}=\frac{1}{2\pi
e^{-\frac{2}{T}}}\int_0^{2\pi}|G_{z}^{\xi}(j)|^2dE.
\] From such relation it follows that
\begin{eqnarray*}
L_{\xi}^A(T) &=&
\sum_{m=0}^{\infty}\frac{2}{T}e^{-\frac{2m}{T}}M_{\xi}^A(m)\\
&=&
\sum_{j=1}^{\infty}\lambda_j\sum_{m=0}^{\infty}\frac{2}{T}e^{-\frac{2m}{T}}
|\langle\varphi_j,U_F^m\xi\rangle|^2\\ &=&
\sum_{j=1}^{\infty}\lambda_j\frac{2}{T}\sum_{m=0}^{\infty}\Big|e^{-\frac{m}{T}}
\int_0^{2\pi}e^{-iE'm}d\mu_j(E')\Big|^2\\ &=&
\sum_{j=1}^{\infty}\lambda_j\frac{2}{T}\,\left\|a^{(j)}\right\|^2_{l^2(\Z)}\\
&=&
\sum_{j=1}^{\infty}\lambda_j\frac{2}{T}\,\left\|\mathcal{F}^{-1}a^{(j)}\right\|^2_{\LL^2[0,2\pi]}\\
&=& \frac{1}{\pi
e^{-\frac{2}{T}}}\frac{1}{T}\sum_{j=1}^{\infty}\lambda_j\int_0^{2\pi}\left|G_{z}^{\xi}(j)\right|^2dE,
\end{eqnarray*}
which is exactly the stated result.
\end{proof}

\

Theorem \ref{FormulaTheorem} clearly remains true if the
eigenvalues $\lambda_j$ of $A$ have finite multiplicity. In this
case, for each $\lambda_j$ consider the corresponding orthonormal
eigenvectors $\varphi_{j_1},\cdots,\varphi_{j_k}$, and one obtains
\[L_{\xi}^A(T)=\frac{1}{\pi e^{-\frac{2}{T}}}
\frac{1}{T}\sum_{j=1}^{\infty}\lambda_j\left(\sum_{n=1}^k\int_0^{2\pi}
\left|\langle\varphi_{j_n},R_z(U_F)\xi\rangle\right|^2dE\right),
\] with $z$ as before.

\

In case the initial condition is $\xi=\varphi_1$, put $\eta^{(z)}
:=  R_z(U_F)\varphi_1$. Thus, $(U_F-z)\eta^{(z)}=\varphi_1$ and so
$U_F\eta^{(z)}=z\eta^{(z)}+\varphi_1$. Hence
\[\langle\varphi_j,U_F\eta^{(z)}\rangle=z\langle\varphi_j,\eta^{(z)}\rangle+\delta_{j,1}\]
and  by denoting
\[G_z(j) :=  G_z^{\varphi_1}(j),
\] one concludes

\

\begin{Lemma}\label{EquationLemma} \[G_z(j)=\left\{\begin{array}{ccc}
\frac{1}{z}\left(\langle\varphi_1,U_F\eta^{(z)}\rangle-1\right),&
\mathrm{if} & j=1\\
\frac{1}{z}\langle\varphi_j,U_F\eta^{(z)}\rangle,& \mathrm{if} &
j>1
\end{array}\right..\]
\end{Lemma}

\

In Section~\ref{ApplicationsSection} we discuss some Floquet
operators that are known in the literature and analyze their Green
functions through the equation
\[(U_F-z\Id)\eta^{(z)}=\varphi_1.\]

\

\section{Applications}
\label{ApplicationsSection} This section is devoted to some
applications of the formula obtained in
Theorem~\ref{FormulaTheorem}. In general it is not trivial to get expressions
and/or bounds for the Green functions of Floquet operators, so one
of the main goals of the applications that follow are to
illustrate how to approach the method we have just proposed.

\

\subsection{Time-Independent Hamiltonians}\label{subsectTIH}
\label{AutonomousSection} As a first example and illustration of
the formula proposed in Theorem~\ref{FormulaTheorem}, we consider
the special case of autonomous hamiltonians. In this case
$H(t)=H_0$ for all $t$ and we assume that $H_0$ is a positive,
unbounded, self-adjoint operator and with simple discrete
spectrum, $H_0\varphi_j=\chi_j\varphi_j$, so that
$\{\varphi_j\}_{j=1}^{\infty}$ is an orthonormal basis of
$\mathcal{H}$ and $0\leq\chi_1<\chi_2<\chi_3<\cdots$ with
$\chi_j\rightarrow\infty$. For $q>0$ we can consider $H_0^q$ as
our  abstract energy operator $A$, so that its eigenvalues are
$\lambda_j=\chi^q_j$ (since $A$ and $H_0$ have the same
eigenfunctions, we are justified in using the notation $\varphi_j$
for the eigenfunctions of $H_0$). We take $U_F=e^{-iH_0}$ (time
$t=1$) and for $\xi\in \mathcal H$
\[
G_z^\xi(j) = \left\langle \varphi_j , R_z(H_0)\xi \right\rangle =
\left\langle R_{\overline{z}}(H_0)\varphi_j , \xi \right\rangle
=\frac{\langle\varphi_j,\xi\rangle}{e^{-i\chi_j}-z}.
\]
Since $\dom H_0^{{q}}$ is invariant under the time evolution
$e^{-itH_0}$, then for $z=e^{-iE}e^{\frac{1}{T}}$ and $\xi\in\dom
H_0^{{q}}$,  we have
\begin{eqnarray}\label{EnergyEquation}
L_{\xi}^q(T) :=  L_{\xi}^{H_0^q}(T) &=& \frac{1}{\pi
e^{-\frac{2}{T}}}
\frac{1}{T}\sum_{j=1}^{\infty}\chi_j^q\int_0^{2\pi}\left|G_z^{\xi}(j)\right|^2dE
\nonumber\\
&=& \frac{1}{\pi e^{-\frac{2}{T}}}
\frac{1}{T}\sum_{j=1}^{\infty}\chi_j^q\int_0^{2\pi}
\frac{\left|\langle\varphi_j,\xi\rangle\right|^2}{|e^{-i\chi_j}-z|^2}dE\\\nonumber
&=& \frac{1}{\pi e^{-\frac{2}{T}}}
\frac{1}{T}\sum_{j=1}^{\infty}\chi_j^q\left|\langle\varphi_j,\xi\rangle\right|^2\int_0^{2\pi}
\frac{dE}{|e^{-i\chi_j}-z|^2}.
\end{eqnarray}

Thus we need to calculate the integral $I_j := \int_0^{2\pi}
\frac{dE}{|e^{-i\chi_j}-z|^2}$. Let $\gamma$ be the closed path in
$\C$ given by $\gamma(E)=e^{iE}$ with $0\leq E\leq2\pi$,
$\alpha_j=e^{\frac{1}{T}}e^{i\chi_j}$ and
$\beta_j=e^{-\frac{1}{T}}e^{i\chi_j}$, then
\begin{eqnarray*}
I_j &=& \int_0^{2\pi}
\frac{dE}{(e^{-i\chi_j}-z)(e^{i\chi_j}-\overline{z})}\\ &=&
\int_0^{2\pi}
\frac{dE}{(e^{-i\chi_j}-e^{-iE}e^{\frac{1}{T}})(e^{i\chi_j}-e^{iE}e^{\frac{1}{T}})}\\
&=& \int_0^{2\pi}
\frac{dE}{e^{\frac{2}{T}}(e^{-\frac{1}{T}}e^{-i\chi_j}-e^{-iE})
(e^{-\frac{1}{T}}e^{i\chi_j}-e^{iE})}\\ &=&
-\frac{1}{e^{\frac{2}{T}}}\int_0^{2\pi}\frac{dE}{e^{-iE}e^{-\frac{1}{T}}e^{-i\chi_j}
(e^{iE}-\alpha_j)(e^{iE}-\beta_j)}\\ &=&
-\frac{1}{e^{\frac{1}{T}}e^{-i\chi_j}}\frac{1}{i}\int_0^{2\pi}\frac{ie^{iE}dE}
{(e^{iE}-\alpha_j)(e^{iE}-\beta_j)}\\ &=&
\frac{i}{e^{\frac{1}{T}}e^{-i\chi_j}}\int_{\gamma}\frac{dw}{(w-\alpha_j)(w-\beta_j)}.
\end{eqnarray*}
As $|\alpha_j|>1$ and $|\beta_j|<1$, $\beta_j$ is the unique pole
in the interior of $\gamma$. Thus, by using residues,
\[I_j=\frac{i}{e^{\frac{1}{T}} e^{-i\chi_j}}2\pi
i\frac{1}{(\beta_j-\alpha_j)}=\frac{2\pi}{e^{\frac{2}{T}}-1}\] and
$I_j$ is independent of $\chi_j$.

Therefore by (\ref{EnergyEquation}) it follows that
\begin{eqnarray*}
L_{\xi}^q(T) &=& \frac{1}{\pi e^{-\frac{2}{T}}}
\frac{1}{T}\sum_{j=1}^{\infty}\chi_j^q\left|\langle\varphi_j,\xi\rangle\right|^2
\frac{2\pi}{e^{\frac{2}{T}}-1}\\ &=& \frac{2}{e^{-\frac{2}{T}}}
\frac{1}{T}\frac{1}{\left(e^{\frac{2}{T}}-1\right)}
\sum_{j=1}^{\infty}\chi_j^q\left|\langle\varphi_j,\xi\rangle\right|^2\\
&=&
\frac{2}{\left(1-e^{-\frac{2}{T}}\right)}\frac{1}{T}\left\|H_0^{\frac{q}{2}}\xi\right\|^2.
\end{eqnarray*}

Since
$\left(1-e^{-\frac{2}{T}}\right)=\frac{2}{T}+\mathcal{O}(\frac{1}{T^2})$,
 for large $T$ it is found that
\[L_{\xi}^q(T)\approx\left\|H_0^{\frac{q}{2}}\xi\right\|^2,\]with (for
$\xi\in\dom H^q_0$)
\[
\lim_{T\to\infty} L_{\xi}^q(T)=\langle\xi,H_0^q\xi\rangle.
\] Then we
conclude that the function \[ \N\ni m\mapsto \left\langle
e^{-iH_0m}\xi,H_0^qe^{-iH_0m}\xi\right\rangle
\] is bounded for
$\xi\in\dom H_0^q$, which is (of course) an expected result (see
Proposition~\ref{EqualityEquation}).

\subsection{Lower Bounded Green Functions} As a first theoretical
application we get dynamical instability from some lower bounds of
the Green functions. See \cite{DST} for a similar result in the
one-dimensional discrete Schr\"o\-dinger operators context; there,
a relation to transfer matrices allows interesting applications to
nontrivial models, what is not available in the unitary setting
yet (and it constitutes of an important open problem). As before,
$\lambda_j$ denote the increasing sequence of positive eigenvalues
of the abstract energy operator $A$, the ones we use to probe
(in)stability.

 Let $[\cdot]$ denotes the integer part of a real number and $|\cdot|$
indicates Lebesgue measure.

\begin{Theorem}\label{ApplTheorem} Suppose that there exist $K>0$ and
$\alpha>0$ such that for each $2N>0$ large enough  there exists a
nonempty Borel set $J(N)\subset S^1$ such that
\[\left|G_z^{\xi}(j)\right|\geq\frac{K}{N^{\alpha}},\qquad N\leq
j\leq2N,\]  holds for all $z=e^{-iE+\frac{1}{T}}$ with $E\in
J_T(N)=\{E''\in S^1:\exists\;E'\in J(N);|E''-E'|\leq\frac{1}{T}\}$
(the $\frac{1}{T}-$neighborhood de $J(N)$). Let $\delta>0$; then
for~$T$ large enough such that $N(T)=[T^\delta]$, one has
\[L_{\xi}^A(T)\geq\cte\lambda_{[T^{\delta}]}T^{\delta(1-2\alpha)-2}.\]
Moreover, if $\lambda_j\geq\cte j^{\gamma}$, $\gamma\ge 0$,  then
\[L_{\xi}^A(T)\geq\cte T^{\delta(\gamma-2\alpha+1)-2}.\]
\end{Theorem}
\begin{proof} By the formula in Theorem~\ref{FormulaTheorem}, or its more
general version~(\ref{formulaGeneR}),
\begin{eqnarray*}
L_{\xi}^{A}(T) &\ge& \frac{1}{\pi
e^{-\frac{2}{T}}}\frac{1}{T}\sum_{j=1}^{\infty}\lambda_j\int_0^{2\pi}
\left|G_z^{\xi}(j)\right|^2dE\\ &\geq&
\frac{\cte}{T}\sum_{j=N(T)}^{2N(T)}\lambda_j\int_0^{2\pi}
\left|G_z^{\xi}(j)\right|^2dE\\ &\geq&
\frac{\cte}{T}\lambda_{N(T)}\sum_{j=N(T)}^{2N(T)}\int_{J_T(N)}\left|G_z^{\xi}(j)\right|^2dE\\
&\geq&
\frac{\cte}{T}\lambda_{N(T)}\sum_{j=N(T)}^{2N(T)}\frac{K^2}{N(T)^{2\alpha}}|J_T(N)|\\
&=& \frac{\cte}{T}|J_T(N)|\lambda_{N(T)}\frac{K^2}{N(T)^{2\alpha-1}}\\
&=&
\frac{\cte}{T}|J_T(N)|\lambda_{[T^{\delta}]}\frac{1}{[T^{\delta}]^{2\alpha-1}}\\
&\geq& \cte\lambda_{[T^{\delta}]}T^{\delta(1-2\alpha)-2};
\end{eqnarray*}
we have used that $|J_T(N)|\geq\frac{1}{T}$.
If $\lambda_j\geq\cte j^{\gamma}$ then
\[L_{\xi}^A(T)\geq\cte T^{\delta\gamma}T^{\delta(1-2\alpha)-2}=
\cte T^{\delta(\gamma-2\alpha+1)-2}.
\] The proof is complete.
\end{proof}

The above theorem becomes appealing when the exponent of $T$
 is greater than zero and instability is obtained, for instance when $\delta(\gamma-2\alpha+1)>2$
in case $\lambda_j\geq\cte j^{\gamma}$. However, up to now we have
not yet been able to find explicit estimates in models of
interest; in any event, we think it will be useful the future
applications and so we point out some speculations. First, note
that it applies even if the set $J(N)$ is a single point!
Nevertheless, we expect that Theorem~\ref{ApplTheorem} will be
applied to models whose Floquet operators have some kind of
``fractal spectrum'' (usually singular continuous or uniformly
H\"older continuous spectral measures) and, somehow, $\alpha$
should be related to dimensional properties of those spectra;
indeed, this was our first motivation for the derivation of this
result and, in our opinion, such applications are among the most
interesting open problems left here. 

\

\subsection{Rank-One Kicked Perturbations}
\label{RankSection} Now consider \[H(t)=H_0+\kappa
P_{\phi}\sum_n\delta(t-n2\pi),
\] with $H_0$   as in Subsection~\ref{AutonomousSection}, with eigenvectors
$\{\varphi_j\}_{j=1}^{\infty}$ and $\chi_j$ the corresponding
eigenvalues; $P_{\phi}(\cdot)=\langle\phi,\cdot\rangle\phi$ where
$\kappa\in\R$ and $\phi$ is a normalized cyclic vector for $H_0$, in the sense that $\|\phi\|=1$ and the closed subspace spanned by
$\{H_0^m\phi:m\in\N\}$ equals $\hil$. Let
\[\phi=\sum_jb_j\varphi_j.\] In this case (see~\cite{C,B})
\[
U_F=U_0\left(\Id+\alpha P_{\phi}\right),
\] with $U_0=e^{-i2\pi H_0}$ and $\alpha=(e^{-i2\pi\kappa}-1)$. Note that
$\phi\in\dom H_0^q$, $\forall q>0$, and so for $\xi\in\dom H_0^q$,
\[
U_F\xi = U_0\xi + \alpha  \langle \phi,\xi\rangle \,U_0\phi
\]  also belongs to $\dom H_0^q$; a simple iteration process shows that
$U_F^m\xi\in \dom H_0^q$ for all $m\ge0$ and we are justified in
using the formula in Theorem~\ref{FormulaTheorem} to estimate
Laplace averages.

We are interested in $\eta^{(z)}=R_z(U_F)\varphi_1$. As $|z|\neq1$
it follows that $\eta^{(z)}$ belongs to the Hilbert space and so
one can write
\[\eta^{(z)}=\sum_{j=1}^{\infty}a_j\varphi_j.\] Note that
$a_j=G_z(j)$ and we have
\begin{equation}\label{RankEquation}
U_F\eta^{(z)}-z\eta^{(z)}=\varphi_1.
\end{equation}
By the relation
\begin{eqnarray*}
U_F\eta^{(z)} &=& U_0\eta^{(z)}+\alpha U_0P_{\phi}\eta^{(z)}\\ &=&
\sum_{j=1}^{\infty}a_jU_0\varphi_j+\alpha
U_0\langle\phi,\eta^{(z)}\rangle\phi\\ &=&
\sum_{j=1}^{\infty}a_je^{-i2\pi\chi_j}\varphi_j+
\alpha\langle\phi,\eta^{(z)}\rangle\sum_{j=1}^{\infty}b_je^{-i2\pi\chi_j}\varphi_j\\
&=& \sum_{j=1}^{\infty}(a_j+\alpha\langle\phi,\eta^{(z)}\rangle
b_j)e^{-i2\pi\chi_j}\varphi_j,
\end{eqnarray*}
 and (\ref{RankEquation}) it follows that
\[\sum_{j=1}^{\infty}(a_j+\alpha\langle\phi,\eta^{(z)}\rangle
b_j)e^{-i2\pi\chi_j}\varphi_j
-z\sum_{j=1}^{\infty}a_j\varphi_j=\varphi_1,\] that is,
\[\sum_{j=1}^{\infty}\left[a_j(e^{-i2\pi\chi_j}-z)+\alpha\langle\phi,\eta^{(z)}\rangle
b_je^{-i2\pi\chi_j}\right]\varphi_j=\varphi_1,\] and we get the
equations
\begin{eqnarray*}a_1(e^{-i2\pi\chi_1}-z)+\alpha\langle\phi,\eta^{(z)}\rangle
b_1e^{-i2\pi\chi_1}&=&1,\\
a_j(e^{-i2\pi\chi_j}-z)+\alpha\langle\phi,\eta^{(z)}\rangle
b_je^{-i2\pi\chi_j} &=& 0\qquad\text{for}\;\; j>1.
\end{eqnarray*} Thus
\begin{equation}\label{RankEquation1}
a_1=\frac{1-\alpha\langle\phi,\eta^{(z)}\rangle
b_1e^{-i2\pi\chi_1}}{e^{-i2\pi\chi_1}-z},
\end{equation}
\begin{equation}\label{rankEquation2}
a_j=-\frac{\alpha\langle\phi,\eta^{(z)}\rangle
b_je^{-i2\pi\chi_j}}{e^{-i2\pi\chi_j}-z}, \qquad j>1.
\end{equation}

For the trivial case $\alpha=0$ or, equivalently, $\kappa\in\Z$,
one has
\[a_1=\frac{1}{e^{-i2\pi\chi_1}-z},
\]
\[
a_j=0,\qquad j>1,
\] and $\eta^{(z)}=\frac{\varphi_1}{e^{-i2\pi\chi_1}-z}$. In
this case the analysis of $L_{\varphi_1}^q(T)$ is reduced to
\[\int_0^{2\pi}|a_1|^2dE=\int_0^{2\pi}\frac{dE}{|e^{-i2\pi\chi_1}-z|^2}=
\frac{2\pi}{e^{\frac{2}{T}}-1}\] as calculate in
Subsection~\ref{AutonomousSection}. Thus
$L_{\varphi_1}^q(T)\approx\|H_0^q\varphi_1\|$ for large~$T$, as
expected.

Returning to the general case $\alpha\neq0$, note that
\begin{eqnarray*}
\langle\phi,\eta^{(z)}\rangle &=&
\sum_{j=1}^{\infty}\overline{b}_ja_j\\ &=&
\overline{b}_1\Big(\frac{1-\alpha\langle\phi,\eta^{(z)}\rangle
b_1e^{-i2\pi\chi_1}}{e^{-i2\pi\chi_1}-z}\Big)+\sum_{j=2}^{\infty}
\overline{b}_j\frac{(-\alpha)\langle\phi,\eta^{(z)}\rangle
b_je^{-i2\pi\chi_j}}{e^{-i2\pi\chi_j}-z}\\ &=&
\frac{\overline{b}_1}{e^{-i2\pi\chi_1}-z}-
\langle\phi,\eta^{(z)}\rangle\sum_{j=1}^{\infty}
\frac{\alpha|b_j|^2e^{-i2\pi\chi_j}}{e^{-i2\pi\chi_j}-z}.
\end{eqnarray*}
So
\[\langle\phi,\eta^{(z)}\rangle=\frac{\overline{b}_1}{(e^{-i2\pi\chi_1}-z)}
\Big[1+\sum_{j=1}^{\infty}
\frac{\alpha|b_j|^2e^{-i2\pi\chi_j}}{e^{-i2\pi\chi_j}-z}\Big]^{-1}.\]
By  denoting \[\tau(z)=1+\sum_{j=1}^{\infty}
\frac{\alpha|b_j|^2e^{-i2\pi\chi_j}}{e^{-i2\pi\chi_j}-z},\] by
(\ref{RankEquation1}) and (\ref{rankEquation2}) we finally obtain
the relations
\[a_1=\frac{1}{e^{-i2\pi\chi_1}-z}-\frac{\alpha|b_1|^2
e^{-i2\pi\chi_1}\tau(z)^{-1}}{(e^{-i2\pi\chi_1}-z)^2},\]
\[a_j=-\frac{\alpha b_j\overline{b}_1e^{-i2\pi\chi_j}\tau(z)^{-1}}
{(e^{-i2\pi\chi_1}-z)(e^{-i2\pi\chi_j}-z)},\qquad j>1.\]

\subsubsection{A Harmonic Oscillator}Now we present an application of the
above relations to a kicked harmonic oscillator with natural
frequency equals to 1; we will write $L_{\xi}^q=
L_{\xi}^{H_0^q}$.

\begin{Proposition}\label{propHOw1}
Let $H_0$ be a harmonic oscillator hamiltonian with appropriate
parameters so that its eigenvalues are integers $j$, $j\ge1$, and
$U_F=U_0(\Id+\alpha P_{\phi})$ as above. Then for any
$\kappa\in\R$ and cyclic vector $\phi$ for $H_0$, there exists
$C>0$ so that, for $T$ large enough,
\[
L_{\varphi_1}^q(T)\le C,
\] where $\varphi_1$ is the harmonic oscillator ground state. Hence we have $H_0^q$-dynamical stability.
\end{Proposition}
\begin{proof} We use the above notation; note that $\varphi_1\in\dom
H_0^q$, $\forall q>0$ and Theorem~\ref{FormulaTheorem} can be
applied. In this case we have
\[\tau(z)=1+\sum_{j=1}^{\infty}\frac{\alpha|b_j|^2}{1-z}=1+
\frac{\alpha}{1-z}\|\phi\|^2=\frac{1-z+\alpha}{1-z},\] and so
\[a_1=\frac{1}{1-z}-\frac{\alpha|b_1|^2}{(1-z)(e^{-i2\pi\kappa}-z)},\]
\[a_j=-\frac{\alpha b_j\overline{b}_1}{(1-z)(e^{-i2\pi\kappa}-z)},\qquad
j>1.
\] Now we evaluate $I_j := \int_0^{2\pi}|a_j|^2dE$. For
$j>1$ and $\gamma(E)=e^{iE}$, $0\leq E\leq2\pi$,
\begin{eqnarray*}
\int_0^{2\pi}|a_j|^2dE &=& \int_0^{2\pi}\Big|\frac{\alpha
b_j\overline{b}_1}{(1-z)(e^{-i2\pi\kappa}-z)}\Big|^2dE\\ &=&
|\alpha|^2|b_j|^2|\overline{b}_1|^2\int_0^{2\pi}\frac{dE}
{\left|(1-e^{-iE}e^{\frac{1}{T}})(e^{-i2\pi\kappa}-e^{-iE}e^{\frac{1}{T}})\right|^2}\\
&=&
\frac{|\alpha|^2|b_j|^2|\overline{b}_1|^2}{ie^{\frac{2}{T}}e^{-i2\pi\kappa}}
\int_{\gamma}\frac{wdw}{(w-\beta_1)(w-\beta_2)
(w-\beta_3)(w-\beta_4)},
\end{eqnarray*}
where $\beta_1=e^{\frac{1}{T}}$, $\beta_2=e^{-\frac{1}{T}}$,
$\beta_3=e^{\frac{1}{T}}e^{i2\pi\kappa}$ and
$\beta_4=e^{-\frac{1}{T}}e^{i2\pi\kappa}$; only  $\beta_2$ and
$\beta_4$ are poles in the interior of $\gamma$. By residue, for
$j>1$,
\begin{eqnarray*} I_j &=&
\frac{2\pi|\alpha|^2|b_j|^2|\overline{b}_1|^2}
{e^{\frac{2}{T}}e^{-i2\pi\kappa}}\times\\ & &
\left(\frac{\beta_2}{(\beta_2-\beta_1)
(\beta_2-\beta_3)(\beta_2-\beta_4)}
+\frac{\beta_4}{(\beta_4-\beta_1)
(\beta_4-\beta_2)(\beta_4-\beta_3)}\right)\\ &=&
\frac{2\pi\alpha|b_j|^2|\overline{b_1}|^2}{(e^{\frac{2}{T}}-1)
(e^{-i2\pi\kappa}-e^{\frac{2}{T}})}-
\frac{2\pi\alpha|b_j|^2|\overline{b_1}|^2e^{i2\pi\kappa}}{(e^{\frac{2}{T}}-1)
(e^{i2\pi\kappa}-e^{\frac{2}{T}})}\\
&=&
\frac{2\pi\alpha|b_j|^2|\overline{b}_1|^2}{(e^{\frac{2}{T}}-1)}\left(\frac{1}
{e^{-i2\pi\kappa}-e^{\frac{2}{T}}}-\frac{e^{i2\pi\kappa}}
{e^{i2\pi\kappa}-e^{\frac{2}{T}}}\right),
\end{eqnarray*}
and for $j=1$
\begin{eqnarray*}
I_1 &=&
\int_0^{2\pi}\Big|\frac{1}{1-z}-\frac{\alpha|b_1|^2}{(1-z)(e^{-i2\pi\kappa}-z)}\Big|^2dE\\
&=& \int_0^{2\pi}\frac{dE}{(1-z)(1-\overline{z})}-
\overline{\alpha}|b_1|^2\int_0^{2\pi}\frac{dE}{(1-z)(1-\overline{z})
(e^{i2\pi\kappa}-\overline{z})}\\ & &
-\alpha|b_1|^2\int_0^{2\pi}\frac{dE}{(1-z)(1-\overline{z})
(e^{-i2\pi\kappa}-z)}\\ & &
+|\alpha|^2|b_1|^4\int_0^{2\pi}\frac{dE}{(1-z)(1-\overline{z})
(e^{-i2\pi\kappa}-z)(e^{i2\pi\kappa}-\overline{z})};
\end{eqnarray*}
evaluating the integrals we obtain
\begin{eqnarray*}
I_1 &=& \frac{2\pi}{(e^{\frac{2}{T}}-1)}
-\frac{2\pi|b_1|^2}{(e^{\frac{2}{T}}-1)}
-\frac{2\pi|b_1|^2}{(e^{i2\pi\kappa}-e^{\frac{2}{T}})}
-\frac{2\pi\alpha|b_1|^2}{(e^{\frac{2}{T}}-1)(e^{-i2\pi\kappa}-e^{\frac{2}{T}})}\\
& & + \frac{2\pi\alpha|b_1|^4}{(e^{\frac{2}{T}}-1)}\left(\frac{1}
{e^{-i2\pi\kappa}-e^{\frac{2}{T}}}-\frac{e^{i2\pi\kappa}}
{e^{i2\pi\kappa}-e^{\frac{2}{T}}}\right),
\end{eqnarray*}
and after inserting this in the expression of the average energy
we get
\begin{eqnarray*}
L_{\varphi_1}^q(T) &=&
\frac{2}{(1-e^{-\frac{2}{T}})T}\left(1-|b_1|^2-\frac{\alpha|b_1|^2}
{(e^{-i2\pi\kappa}-e^{\frac{2}{T}})}\right)\\ & &
-\frac{2|b_1|^2}{e^{-\frac{2}{T}}(e^{i2\pi\kappa}-e^{\frac{2}{T}})T}\\
& & +\frac{2\alpha|b_1|^2}{(1-e^{-\frac{2}{T}})T} \left(\frac{1}
{e^{-i2\pi\kappa}-e^{\frac{2}{T}}}-\frac{e^{i2\pi\kappa}}
{e^{i2\pi\kappa}-e^{\frac{2}{T}}}\right)\langle\phi,H_0^q\phi\rangle.
\end{eqnarray*}
Therefore, for large $T$ there is a constant $C(\kappa,b_1)>0$ so
that
\[L_{\varphi_1}^q(T)\leq C(\kappa,b_1)\left(1+\langle\phi,H_0^q\phi\rangle
+\frac{1}{T}\right).\] This completes the proof.
\end{proof}

For harmonic oscillators with eigenvalues $\omega j$,
$\omega\ne1$, the evaluations of the resulting integrals are more
intricate and were not carried out.

\

\subsection{Kicked Perturbations by a $V$ in $\LL^2(S^1)$}
\label{KickedSection}
\subsubsection{Kicked Linear Rotor}Consider \[H(t)=\omega
p+V(x)\sum_{n\in\Z}\delta(t-n2\pi),
\] where $p=-i\frac{d}{dx}$,
$\omega\in\R$ and $V\in \LL^2(S^1)$. The Hilbert space is
$\LL^2(S^1)$; this model was considered in
\cite{Bell,dBF,deOlinear} and references therein. The Floquet
operator is
\[U_F=U_V=e^{-i2\pi\omega p}e^{-iV(x)}.\] Denote by
$\varphi_j(x)=e^{ijx}/\sqrt{2\pi}$, $0\le x<2\pi$ and $j\in\Z$, be
the eigenvectors of $p^2$ whose eigenvalues are the square of
integers $j^2$; all eigenvalues have multiplicity 2 (the
corresponding eigenvectors are $\varphi_j$ and $\varphi_{-j}$),
except for the null eigenvalue which is simple.

Consider the case $\omega=1$; then
\[\left((U_F-z)^{-1}\varphi_0\right)(x)=\frac{1}{\sqrt{2\pi}(e^{-iV(x)}-z)},
\] and so
\[G_z^{\varphi_0}(j)=\langle\varphi_j,R_z(U_F)\varphi_0\rangle=\frac{1}{2\pi}
\int_0^{2\pi}\frac{e^{-ijx}}{e^{-iV(x)}-z}dx.\] Denote $I_j :=
\int_0^{2\pi}|G_z^{\varphi_0}(j)|^2dE$. It follows that
\begin{eqnarray*}
I_j &=& \frac{1}{(2\pi)^2}
\int_0^{2\pi}\Big|\int_0^{2\pi}\frac{e^{-ijx}}{e^{-iV(x)}-z}dx\Big|^2dE \\
\\ &=& \frac{1}{(2\pi)^2}
\int_0^{2\pi}\int_0^{2\pi}e^{-ijx}e^{ijy}\left(\int_0^{2\pi}\frac{dE}
{(e^{-iV(x)}-z)(e^{iV(y)}-\overline{z})}\right)dxdy.
\end{eqnarray*}
For $x,y\in S^1$ fixed denote $I_{xy} := \int_0^{2\pi}\frac{dE}
{(e^{-iV(x)}-z)(e^{iV(y)}-\overline{z})}$. If $\gamma(E)=e^{iE}$,
$0\leq E\leq2\pi$, one has
\begin{eqnarray*}
I_{xy} &=& \int_0^{2\pi}\frac{dE}
{(e^{-iV(x)}-e^{-iE}e^{\frac{1}{T}})(e^{iV(y)}-e^{iE}e^{\frac{1}{T}})}\\
&=& \int_0^{2\pi}\frac{dE}
{e^{-iE}e^{-iV(x)}(e^{iE}-e^{iV(x)}e^{\frac{1}{T}})e^{\frac{1}{T}}
(e^{-\frac{1}{T}}e^{iV(y)}-e^{iE})}\\ &=&
-\frac{1}{e^{\frac{1}{T}}e^{-iV(x)}}\frac{1}{i}\int_{\gamma}\frac{dw}
{(w-e^{iV(x)}e^{\frac{1}{T}})(w-e^{-\frac{1}{T}}e^{iV(y)})},
\end{eqnarray*}
and by residues
\[I_{xy}=-\frac{2\pi}{e^{\frac{1}{T}}e^{-iV(x)}(e^{-\frac{1}{T}}e^{iV(y)} -
e^{iV(x)}e^{\frac{1}{T}})}=\frac{2\pi}{(e^{\frac{2}{T}}-e^{-iV(x)}e^{iV(y)})}.\]
Hence
\begin{eqnarray}\label{rankEquation3}
I_j &=& \frac{1}{(2\pi)^2}
\int_0^{2\pi}\int_0^{2\pi}e^{-ijx}e^{ijy}
\frac{2\pi}{(e^{\frac{2}{T}}-e^{-iV(x)}e^{iV(y)})}dxdy\nonumber\\
&=& \frac{1}{2\pi}
\int_0^{2\pi}e^{-ijx}\Big(\int_0^{2\pi}\frac{e^{ijy}dy}
{(e^{\frac{2}{T}}-e^{-iV(x)}e^{iV(y)})}\Big)dx\\ &=&
\frac{1}{2\pi}
\int_0^{2\pi}\frac{e^{-ijx}}{e^{-iV(x)}}\Big(\int_0^{2\pi}\frac{e^{ijy}dy}
{(e^{\frac{2}{T}}e^{iV(x)}-e^{iV(y)})}\Big)dx. \nonumber
\end{eqnarray}

The analytical evaluation of these integrals is not a simple task.   As an
illustration, consider the particular potential $V(x)=x$; since by
Cauchy's integral formula
\[\int_0^{2\pi}\frac{e^{ijy}dy}
{(e^{\frac{2}{T}}e^{ix}-e^{iy})}=-\frac{1}{i}
\int_{\gamma}\frac{w^{j-1}dw}{(w-e^{\frac{2}{T}}e^{ix})}=0,\quad\text{if}\;\;j\geq1,\]
and by residue theorem
\[\int_0^{2\pi}\frac{e^{ijy}dy}
{(e^{\frac{2}{T}}e^{ix}-e^{iy})}= -\frac{1}{i}
\int_{\gamma}\frac{dw}{w^{1-j}(w-e^{\frac{2}{T}}e^{ix})}=
\frac{2\pi}{(e^{\frac{2}{T}}e^{ix})^{1-j}},\quad\text{if}\;\;j\leq0,
\]
it is found that
\[
I_j=0\qquad\text{if}\;\;j\geq1
\] and
\[
I_j=\frac{1}{2\pi}\int_0^{2\pi}\frac{e^{-ijx}}{e^{-ix}}\frac{2\pi}{(e^{\frac{2}{T}}e^{ix})^{1-j}}dx=
\frac{1}{e^{\frac{2}{T}(1-j)}}\int_0^{2\pi}dx=
\frac{2\pi}{e^{\frac{2}{T}(1-j)}},\quad\text{if}\;\;j\leq0.\]
Therefore, by~(\ref{formulaGeneR}) it follows that for any $q>0$
\[L_{\varphi_0}^{p^{2q}}(T)\ge\frac{1}{\pi
e^{-\frac{2}{T}}}\frac{1}{T}\sum_{j=1}^{\infty}j^{2q}I_{-j}=
\frac{2}{T}\sum_{j=1}^{\infty}j^{2q}e^{-\frac{2}{T}j}\] and we
conclude that (see the Appendix)
\[
L_{\varphi_0}^{p^{2q}}(m)\ge \mathrm{cte}\; m^{2q}
\] and also that the sequence  $m\mapsto \left\langle U_F^m \varphi_0,p^{2q}U_F^m
\varphi_0\right\rangle $ is unbounded. This behavior is expected
since the spectrum of $U_F$ is absolutely continuous in this case
\cite{Bell}, but here we got the result explicitly without passing
through spectral arguments, although in a rather involved way;
indeed, a much simpler derivation is
possible by direct calculating $U_F^m\varphi_0$ and the corresponding expectation values.

For $V(x)=kx$ with integer $k\geq2$,  similar results are
obtained, that is
\[I_j=\left\{\begin{array}{ccc}
0 & \mbox{if} & j=lk,\; l\geq1\\
\frac{2\pi}{e^{{2/T}}(1-l)} & \mbox{if} & j=lk,\;l\leq0
\end{array}\right.,\] and so
\[L_{\varphi_0}^{p^{2q}}(T)\geq\frac{2k^{2q}}{T}
\sum_{l=1}^{\infty}l^{2q}e^{-\frac{2}{T}l}.
\] Therefore, we have the following lower bound for the Laplace average
\[
L_{\varphi_0}^{p^{2q}}(m)\ge C(k,q)\,m^{2q}
\] (see Appendix). The
same is valid if $V(x)=kx$ with $k$ denoting any negative integer
number.

\subsubsection{Power Kicked Systems}
Due to the difficulty in evaluating the integrals in
(\ref{rankEquation3}), in order to estimate
$L_{\varphi_0}^{p^{2q}}(T)$ in some situations we take an
alternative way.

Consider the Kicked models in $\LL^2(S^1)$ with Floquet operator
\begin{equation}\label{eqfpv}U_F=U_V=e^{-i2\pi\omega f(p)}e^{-iV(x)},
\end{equation} corresponding to the
hamiltonian \[H(t)=\omega f(p)+V(x)\sum_{n\in\Z}\delta(t-2\pi n),
\]
with $p,V,\varphi_j$  as before and $f(p)=p^N$ for some $N\in\N$.
Let $\mathcal{F}:\LL^2(S^1)\rightarrow l^2(\Z)$ be the Fourier
transform. Then $\mathcal{F}U_V\mathcal{F}^{-1}:l^2(\Z)\rightarrow
l^2(\Z)$ and
\[\mathcal{F}U_V\mathcal{F}^{-1}=
\mathcal{F}e^{-i2\pi\omega f(p)}e^{-iV(x)}\mathcal{F}^{-1}
=\mathcal{F}e^{-i2\pi\omega
f(p)}\mathcal{F}^{-1}\mathcal{F}e^{-iV(x)}\mathcal{F}^{-1}\] where
$\mathcal{F}e^{-i2\pi\omega f(p)}\mathcal{F}^{-1}$ is represented
by a diagonal matrix $D$ whose elements are
\[D(m,n)=e^{-i2\pi\omega f(n)}\delta_{mn},
\] and
$\mathcal{F}e^{-iV(x)}\mathcal{F}^{-1}$ is represented by a matrix
$W$ whose elements are
\[W(m,n)=(\mathcal{F}\rho)(m-n)=\hat{\rho}(m-n),
\] where
$\rho(x)=\frac{1}{\sqrt{2\pi}}e^{-iV(x)}$. Denote $B=DW$; so
\[
B(m,n)=e^{-i2\pi\omega f(n)}\hat{\rho}(m-n)
\] and
\begin{equation}\label{rankEquation4} U_V=\mathcal{F}^{-1}B\mathcal{F}.
\end{equation}
Put $\eta^{(z)}=R_z(U_V)\varphi_0$; then
\[U_V\eta^{(z)}-z\eta^{(z)}=\varphi_0,
\] and using
(\ref{rankEquation4}) we obtain
\[B\mathcal{F}\eta^{(z)}-z\mathcal{F}\eta^{(z)}=\mathcal{F}\varphi_0.
\]
Thus, for each $n\in\Z$,
\[(B\mathcal{F}\eta^{(z)})(n)-(z\mathcal{F}\eta^{(z)})(n)=(\mathcal{F}\varphi_0)(n),
\]
so that
\begin{equation}\label{rankEquation5}
e^{-i2\pi\omega
f(n)}\sum_{j\in\Z}\hat{\rho}(n-j)G_z^{\varphi_0}(j)-zG_z^{\varphi_0}(n)=\delta_{n0}.
\end{equation}

\

\noindent {\bf Tridiagonal Case} In order to deal with the above
equations, we try to simplify them by supposing that $V$ is such
that $\hat{\rho}(m-n)=0$ if $|m-n|>1$. Then, for each $n\in\Z$
fixed (\ref{rankEquation5}) becomes
\begin{equation}\label{rankEquation6}
e^{-i2\pi\omega
f(n)}\sum_{|n-j|\leq1}\hat{\rho}(n-j)G_z^{\varphi_0}(j)-zG_z^{\varphi_0}(n)=\delta_{n0}
\end{equation} and $\mathcal{F}^{-1}U_V\mathcal{F}=B$ is tridiagonal and
has the structure
\[B=\left(\begin{array}{cccccc}
\ddots & & & & &   \\
       & g(-1)\hat{\rho}(0) & g(-1)\hat{\rho}(-1)& & &  \\
       & \hat{\rho}(1) & \hat{\rho}(0) & \hat{\rho}(-1) & &
        \\ & & g(1)\hat{\rho}(1) & g(1)\hat{\rho}(0)
         & g(1)\hat{\rho}(-1) &
        \\ & & & g(2)\hat{\rho}(1) & g(2)\hat{\rho}(0) & \\ & &
       & &  & \ddots\\
\end{array}\right)\] where $g(n)=e^{-i2\pi\omega
f(n)}$.

Now, a tridiagonal unitary operator $U$ on $l^2(\Z)$
 is either unitarily equivalent to a (bilateral) shift
operator, or it is an infinite direct sum of $2\times2$ and
$1\times1$ unitary matrices, as shown in Lemma~3.1 of \cite{BHJ}.
For proving this result it was only used  that  $U$ is unitary and
$Ue_k=\alpha_ke_{k-1}+\beta_ke_k+\gamma_ke_{k+1}$, where $\{e_k\}$
is the canonical basis of $l^2(\Z)$, that is,
\[U=\left(\begin{array}{ccccc} \ddots & \alpha_{k-1} & & & \\ &
\beta_{k-1} & \alpha_k & & \\ & \gamma_{k-1} & \beta_k &
\alpha_{k+1} & \\ & & \gamma_k & \beta_{k+1} & \\ & & &
\gamma_{k+1} & \ddots
\end{array}\right)\] It then follows that
for all $k\in\Z$
\[|\alpha_k|^2+|\beta_k|^2+|\gamma_k|^2=1,\]
\[\gamma_{k-1}\overline{\beta_{k-1}}+\beta_k\overline{\alpha_k}=0,\]
\[\alpha_k\overline{\gamma_k}=0.\]
Applying these relations to $B=\mathcal{F}^{-1}U_V\mathcal{F}$ we
obtain
\begin{itemize}
\item If $\hat{\rho}(-1)\neq0$ then
$\hat{\rho}(1)=\hat{\rho}(0)=0$ and $|\hat{\rho}(-1)|=1$.
\item If $\hat{\rho}(1)\neq0$ then
$\hat{\rho}(-1)=\hat{\rho}(0)=0$ and $|\hat{\rho}(1)|=1$.
\item If $\hat{\rho}(0)\neq0$ then
$\hat{\rho}(1)=\hat{\rho}(-1)=0$ and $|\hat{\rho}(0)|=1$.
\end{itemize}

The next step is to investigate these cases. If
$\hat{\rho}(0)\neq0$ it reduces to the autonomous case  $H(t)=H_0$
previously considered.

The cases $\hat{\rho}(-1)\neq0$ and $\hat{\rho}(1)\neq0$ are
similar, so we only discuss that $\hat{\rho}(1)\neq0$. For
$n\in\Z$ fixed,  equation (\ref{rankEquation6}) takes the form
\begin{equation}\label{EquationEsq}
e^{-i2\pi\omega
f(n)}\hat{\rho}(1)G_z^{\varphi_0}(n-1)-zG_z^{\varphi_0}(n)=\delta_{n0},
\end{equation}
so we can write $G_z^{\varphi_0}(n)$ in terms of
$G_z^{\varphi_0}(0)$ and $G_z^{\varphi_0}(-1)$ for all $n\in\Z$.
More precisely
\[G_z^{\varphi_0}(n)=\frac{e^{-i2\pi\omega(f(n)+\cdots+f(1))}\hat{\rho}(1)^n}{z^n}
G_z^{\varphi_0}(0)\qquad n\geq1,\]
\[G_z^{\varphi_0}(-n)=\frac{z^{n-1}}
{e^{-i2\pi\omega(f(-n+1)+\cdots+f(-1))}\hat{\rho}(1)^{n-1}}
G_z^{\varphi_0}(-1)\qquad n\geq2;\] moreover, for $n=0$ in
(\ref{EquationEsq}) we obtain
$\hat{\rho}(1)G_z^{\varphi_0}(-1)-zG_z^{\varphi_0}(0)=1$, so for
$z=e^{-iE}e^{1/T}$ and $T>1$
\begin{eqnarray*}
1 &\leq& |G_z^{\varphi_0}(-1)|+|z||G_z^{\varphi_0}(0)|\\ &=&
|G_z^{\varphi_0}(-1)|+e^{1/T}|G_z^{\varphi_0}(0)|\\ &\leq&
e(|G_z^{\varphi_0}(-1)|+|G_z^{\varphi_0}(0)|),
\end{eqnarray*} and there exists $d>0$ so that
 \[|G_z^{\varphi_0}(-1)|^2+|G_z^{\varphi_0}(0)|\geq d>0.\]
Therefore, by~(\ref{formulaGeneR}), for $T>1$ one has
\begin{eqnarray*}
L_{\varphi_0}^{p^{2q}}(T) &\ge& \frac{1}{\pi
e^{-\frac{2}{T}}}\frac{1}{T}\sum_{n=1}^{\infty}n^{2q}
\left(\int_0^{2\pi}|G_z^{\varphi_0}(n)|^2dE+\int_0^{2\pi}|G_z^{\varphi_0}(-n)|^2dE\right)\\
&=& \frac{1}{\pi
e^{-\frac{2}{T}}}\frac{1}{T}\sum_{n=1}^{\infty}n^{2q}
\Bigg(\frac{1}{e^{\frac{2n}{T}}}\int_0^{2\pi}|G_z^{\varphi_0}(0)|^2dE\\
& & +
e^{\frac{2(n-1)}{T}}\int_0^{2\pi}|G_z^{\varphi_0}(-1)|^2dE\Bigg)\\
&\geq& \frac{1}{\pi
e^{-\frac{2}{T}}}\frac{1}{T}\sum_{n=1}^{\infty}n^{2q}e^{-\frac{2n}{T}}
\int_0^{2\pi}\left(|G_z^{\varphi_0}(0)|^2+|G_z^{\varphi_0}(-1)|^2\right)dE\\
&\geq& d\frac{2}{T}\sum_{n=0}^{\infty}
(n+1)^{2q}e^{-\frac{2n}{T}},
\end{eqnarray*}
so that, by the discussion at the end of the Appendix,
\[
L_{\varphi_0}^{p^{2q}}(m)\geq C(m+1)^{2q}
\]and  $\left\langle U_V^m \varphi_0,p^{2q}U_V^m \varphi_0\right\rangle $
is unbounded. Hence we have instability.

\

\noindent {\bf Pentadiagonal Case} Suppose now that $V$ is such
that $\hat{\rho}(m-n)=0$ if $|m-n|>2$. Then for each $n\in\Z$
fixed, equation (\ref{rankEquation5}) becomes
\begin{equation}\label{rankEquation9}
e^{-i2\pi\omega
f(n)}\sum_{|n-j|\leq2}\hat{\rho}(n-j)G_z^{\varphi_0}(j)-zG_z^{\varphi_0}(n)=\delta_{n0},
\end{equation} and $\mathcal{F}^{-1}U_V\mathcal{F}$ is pentadiagonal and
has a structure similar to the corresponding operator in the
previous case, just adding the elements whose distance to the
diagonal is $2$. The elements in the new upper diagonal are
$e^{-i2\pi\omega f(n)}\hat{\rho}(-2)$, and the new lower diagonal
are $e^{-i2\pi\omega f(n)}\hat{\rho}(2)$.

For not repeating the tridiagonal case we suppose that either
$\hat{\rho}(2)$ or $\hat{\rho}(-2)$ is different from zero. If $U$
is a pentadiagonal unitary operator in $l^2(\Z)$, that is,
$Ue_k=\zeta_ke_{k-2}+\alpha_ke_{k-1}+\beta_ke_k+\gamma_ke_{k+1}+\theta_ke_{k+2},$
one gets the matrix representation
\[U=\left(\begin{array}{ccccccc} \ddots & & & & & & \\ &
\beta_{k-2} & \alpha_{k-1} & \zeta_k & & &
\\ & \gamma_{k-2} & \beta_{k-1} & \alpha_k & \zeta_{k+1} & & \\ &
\theta_{k-2} & \gamma_{k-1} & \beta_k & \alpha_{k+1} & \zeta_{k+2}
& \\ & & \theta_{k-1} & \gamma_k & \beta_{k+1} & \alpha_{k+2} & \\
 & & & \theta_k & \gamma_{k+1} & \beta_{k+2} & \\ & & & & & & \ddots
\end{array}\right).\] From this we obtain the following relations, for each
$k\in\Z$,
\[|\zeta_k|^2+|\alpha_k|^2+|\beta_k|^2+|\gamma_k|^2+|\theta_k|^2=1\]
\[\overline{\zeta_k}\alpha_{k-1}+ \overline{\alpha_k}\beta_{k-1}
+\overline{\beta_k}\gamma_{k-1}
+\overline{\gamma_k}\theta_{k-1}=0\]
\[\overline{\beta_{k-1}}\theta_{k-1}+ \overline{\alpha_k}\gamma_k
+\overline{\zeta_{k+1}}\beta_{k+1}=0\]
\[\overline{\alpha_{k-1}}\theta_{k-1}+ \overline{\zeta_k}\gamma_k=0\]
\[\overline{\zeta_k}\theta_k=0.\]
Suppose that $\hat{\rho}(2)\neq0$. The case $\hat{\rho}(-2)\neq0$
is similar. Then by the above relations we obtain
$\hat{\rho}(-2)=\hat{\rho}(-1)=\hat{\rho}(0)=\hat{\rho}(1)=0,$
and so (\ref{rankEquation9}) becomes \[e^{-i2\pi\omega
f(n)}\hat{\rho}(2)G_z^{\varphi_0}(n-2)-zG_z^{\varphi_0}(n)=\delta_{n0}.\]
For $n=0$ one gets
$\hat{\rho}(2)G_z^{\varphi_0}(-2)-zG_z^{\varphi_0}(0)=1$ and
analogously to the previous case
\[\vert G_z^{\varphi_0}(-2)\vert^2+\vert
G_z^{\varphi_0}(0)\vert^2\geq d>0,\] with $z=e^{-iE}e^{1/T}$ and
$T>1$. Since for $n\geq1$
\[G_z^{\varphi_0}(2n)=\frac{e^{-i2\pi\omega(f(2n)+f(2n-2)+\cdots+f(2))}\hat{\rho}(2)^n}
{z^n}G_z^{\varphi_0}(0),\]
\[G_z^{\varphi_0}(-2n)=\frac{z^{n-1}G_z^{\varphi_0}(-2)}{\hat{\rho}(2)^{n-1}
e^{-i2\pi\omega(f(-2(n-1))+\cdots+f(-2))}},\] we obtain
\begin{eqnarray*}
L_{\varphi_0}^{p^{2q}}(T) &\geq& \frac{1}{\pi
e^{-\frac{2}{T}}}\frac{1}{T}\sum_{n=1}^{\infty}(2n)^{2q}
\left(\int_0^{2\pi}|G_z^{\varphi_0}(2n)|^2dE+\int_0^{2\pi}|G_z^{\varphi_0}(-2n)|^2dE\right)\\
&=& \frac{1}{\pi
e^{-\frac{2}{T}}}\frac{1}{T}\sum_{n=1}^{\infty}(2n)^{2q}
\Bigg(\frac{1}{e^{\frac{2n}{T}}}\int_0^{2\pi}|G_z^{\varphi_0}(0)|^2dE\\
& & +
e^{\frac{2(n-1)}{T}}\int_0^{2\pi}|G_z^{\varphi_0}(-2)|^2dE\Bigg)\\
&\geq& \frac{1}{\pi
e^{-\frac{2}{T}}}\frac{1}{T}\sum_{n=1}^{\infty}(2n)^{2q}e^{-\frac{2n}{T}}
\int_0^{2\pi}\left(|G_z^{\varphi_0}(0)|^2+|G_z^{\varphi_0}(-2)|^2\right)dE\\
&\geq& d\frac{2}{m}\sum_{n=0}^{\infty}
(2(n+1))^{2q}e^{-\frac{2n}{T}},
\end{eqnarray*} hence \[
L_{\varphi_0}^{p^{2q}}(T)\geq C(2(m+1))^{2q},
\]and  $\left\langle U_V^m \varphi_0,p^{2q}U_V^m \varphi_0\right\rangle $
is unbounded.

\

\noindent {\bf $N$-diagonal Case} If $V$ satisfies
$\hat{\rho}(m-n)=0$ for $|m-n|>N$, we suppose that either
$\hat{\rho}(N)$ or $\hat{\rho}(-N)$ is different from zero. In
case $\hat{\rho}(N)\neq0$, by unitarity and the structure of
$\mathcal{F}^{-1}U_V\mathcal{F}$ we obtain that
$\hat{\rho}(N-1)=\cdots\hat{\rho}(0)=\hat{\rho}(-1)=\cdots=\hat{\rho}(-N)=0$,
thus (\ref{rankEquation5}) becomes, for each $n\in\Z$,
\[e^{-i2\pi\omega
f(n)}\hat{\rho}(N)G_z^{\varphi_0}(n-N)-zG_z^{\varphi_0}(n)=\delta_{n0},\]
and so \[\vert G_z^{\varphi_0}(-N)\vert^2+\vert
G_z^{\varphi_0}(0)\vert^2\geq d>0,\] with $z=e^{-iE}e^{1/T}$,
$T>1$. Moreover, for $n\geq1$
\[G_z^{\varphi_0}(nN)=\frac{e^{-i2\pi\omega(f(nN)+f((n-1)N)+\cdots+f(N))}
\hat{\rho}(N)^n}{z^n}G_z^{\varphi_0}(0)\]
and
\[G_z^{\varphi_0}(-nN)=\frac{z^{n-1}G_z^{\varphi_0}(-N)}{\hat{\rho}(N)^{n-1}
e^{-i2\pi\omega(f(-N(n-1))+\cdots+f(-N))}}.\] Similarly to the
previous cases we conclude that
\[L_{\varphi_0}^{p^{2q}}(T)\geq d\frac{2}{T}\sum_{n=0}^{\infty}
(N(n+1))^{2q}e^{-\frac{2n}{T}}.\]

Therefore we can stated the following result:

\

\begin{Theorem}\label{ConclusionTheorem} For Kicked systems in $\LL^2(S^1)$
with \[U_V=e^{-i2\pi\omega f(p)}e^{-iV(x)}\] as in (\ref{eqfpv}),
we obtain that $\mathcal{F}U_V\mathcal{F}^{-1}:l^2(\Z)\rightarrow
l^2(\Z)$ is represented by the matrix $B$ with elements
$B(m,n)=e^{-i2\pi\omega f(n)}\hat{\rho}(m-n)$, where
$\rho(x)=(2\pi)^{-\frac{1}{2}}e^{-iV(x)}$. If $V$ satisfies
$\hat{\rho}(m-n)=0$ for $|m-n|>N\in\N^*$ and either
$\hat{\rho}(N)$ or $\hat{\rho}(-N)$ is different from zero, then
$V(x)=\pm Nx+\theta$, for some $\theta\in\R$, and
$\mathcal{F}U_V\mathcal{F}^{-1}$ is unitarily equivalent to $T^N$
(the $N$th power of $T$) where $T$ is the bilateral shift.
Furthermore,
\[L_{\varphi_0}^{p^{2q}}(T)\geq
d\frac{2}{T}\sum_{n=0}^{\infty}(N(n+1))^{2q}e^{-\frac{2n}{T}}.\]
\end{Theorem}
\begin{proof} It is enough to prove that $\mathcal{F}U_V\mathcal{F}^{-1}$
is unitarily equivalent to $T^N$. Suppose that
$\hat{\rho}(N)\neq0$ (the case for $\hat{\rho}(-N)\neq0$ is
similar); then by the above discussion we obtain
\[B(m,n)=\left\{\begin{array}{ccc} 0 & \mbox{if} & m\neq n+N \\
e^{-i2\pi\omega f(n)}\hat{\rho}(N) & \mbox{if} & m=n+N
\end{array}\right.,\] that is, $Be_n=e^{-i2\pi\omega
f(n)}\hat{\rho}(N)e_{n+N}$ where $\{e_n\}$ is the canonical basis
of $l^2(\Z)$. Since $|\hat{\rho}(N)|=1$, write
$\hat{\rho}(N)=e^{-i\theta}$. Let $W$ be the unitary operator
defined by
\[We_n=e^{i\vartheta_n}e_n,\qquad n\in\Z,
\] where $\vartheta_n$ are elements in $[0,2\pi)$. If $\vartheta_n$
satisfies for all $n\in\Z$
\begin{equation}\label{rankEquation8}
\vartheta_{n+N}-\vartheta_n=2\pi\omega f(n)+\theta,
\end{equation} it follows that $W^{-1}BW=T^N$. (\ref{rankEquation8}) is
satisfied taking, for example,
$\vartheta_0=\vartheta_1=\cdots=\vartheta_{N-1}=0$ and the another
$\vartheta_n$ obeying (\ref{rankEquation8}).
\end{proof}

Although Theorem~\ref{ConclusionTheorem}  gives a nice
illustration of the potential applications of our expression for
the Laplace average, since it is one of the few instances that
such average can be explicitly estimated from below, again it can
be derived by more direct methods  and one can also conclude
\cite{Bell} that the spectrum of the corresponding Floquet operators are
absolutely continuous.

\

\section{Conclusions}
\label{ConclusionSection} Although most of our applications of
Theorem~\ref{FormulaTheorem} give expected results (sometimes
known results that can be derived in simpler ways), we believe
that that formula is  interesting and has a potential to be
applied to more sophisticated models as the Fermi accelerator. The
difficulty is to get expressions or estimates for the Green
functions, since calculating the resolvent of an operator is not
always an easy task; sometimes we have the expressions for resolvent operators (e.g., for kicked systems) but the resulting integrals can be too
involved. We have not tried any numerical approach to
formula~(\ref{formulaDoTeor}), which might be useful for some
specific models.

In the case of one-dimensional discrete Schr\"o\-dinger operators,
where the hamiltonian is  $H_V:l^2(\Z)\rightarrow l^2(\Z)$ defined
by
\[
(H_V\xi)(n)=\xi(n+1)+\xi(n-1)+V(n)\xi(n),
\] $V$
a bounded sequence, a similar formula can be handled in some cases
by relating the resolvent $R_{E+\frac{i}{T}}(H_V)$ to
 transfer matrices. Then
adequate upper bounds of such transfer matrices, on some set of
energies $E$,  result in lower estimates for the corresponding
Green functions and then transport properties are obtained  for
interesting models (see~\cite{DST} and references therein).

In~\cite{BHJ}  a class of Floquet operators displaying a
pentadiagonal structure was  introduced; for these models there is
a transfer matrix formalism. However, such transfer matrices  are
too complicated and  analytical estimates seem far from trivial.

Anyway, the technique here is quite general, it asks no particular
regularity of the time-dependence and  can be virtually applied to any time-periodic system as soon as the time evolution is well posed. As
already said, the chief difficulty is related to suitable bounds of
matrix elements of the resolvents of unitary (Floquet) operators,
a task harder than we initially envisaged. Herein we
put forward for consideration the challenge of getting additional
applications for the formula (\ref{formulaDoTeor}) deduced for the Laplace averages, including an application of Theorem~\ref{ApplTheorem} to physical models. It
is also worth mentioning the question left open in
Lemma~\ref{lema.49}, that is, is it true that $\beta^{-}_e=
\beta^{-}_d$?

\

\section*{Appendix: Laplace Transform of Sequences}
\label{AppendixSection} Let $a=(a_n)_{n\in\N}$ be a sequence of
positive real numbers. The Laplace transform of $a$, denoted by
$f_a$, is the function defined by
\begin{equation}\label{SeriesEquation}
f_a(s)=\sum_{n=0}^{\infty}e^{-sn}a(n),
\end{equation} for $s$ in a subset of $\R$.
It will also be denoted by $f_a(s)=\mathcal{L}(a)$.

We say the Laplace transform of $a=(a_n)$  exists if the series in
(\ref{SeriesEquation}) converges for some $s$. For example, if
$a(n)=e^{n^2}$, then the sum in~(\ref{SeriesEquation}) diverges
for all $s\in\R$.

\

\noindent {\bf Examples.}
\begin{enumerate}
\item For the constant sequence $a(n)=1$ it
follows that
\[f_a(s)=\sum_{n=0}^{\infty}e^{-sn}=\frac{1}{(1-e^{-s})},\] for
$s>0$. By using Taylor expansion,  for small $s$ one finds that
$\frac{1}{(1-e^{-s})}\approx\frac{1}{s},$ .
\item Since
$\sum_{n=0}^{\infty}z^n=\frac{1}{1-z}$, for $z\in\C,|z|<1$, it
follows that
\[\sum_{n=0}^{\infty}(n+k)(n+k-1)\cdots(n+1)z^n=\frac{k!}{(1-z)^{k+1}},\]
for $k=1,2,3,\cdots$, and $z$ as above. Thus,  the Laplace
transform of $a^k(n)=(n+k)(n+k-1)\cdots(n+1)$ is
\[f_a^k(s)=\sum_{n=0}^{\infty}e^{-sn}a^k(n)=\frac{k!}{(1-e^{-s})^{k+1}},\qquad
s>0.\]
 For small $s$,
$f_a^k(s)\approx\frac{k!}{s^{k+1}}$.
\end{enumerate}
\

A sequence of complex numbers $a=(a_n)$ is said to be exponential
of order $\sigma_0$ (real) if there exists $M>0$ so that
$|a(n)|\leq Me^{\sigma_0n}$, $\forall n$. That is, $a(n)$ does not
increase faster than $e^{\sigma_0n}$ as $n\rightarrow\infty$. If
$a=(a_n)$ is exponential of order $\sigma_0>0,$ then
\[f_a(s)=\sum_{n=0}^{\infty}e^{-sn}a(n)\] is convergent for any
$s>\sigma_0$.

Let $\mathcal{V}$ denote the set of positive sequences of
exponential order $\sigma_0$. The Laplace transform $\mathcal{L}$
satisfies
\[\mathcal{L}(ca)=c\mathcal{L}(a),\qquad\mathcal{L}(a+b)=\mathcal{L}(a)+\mathcal{L}(b),\]
where $c$ is a positive number and $a$ and $b$ are sequences in
$\mathcal{V}$. Moreover, if $a\in \mathcal{V}$ and
$\mathcal{L}(a)=0,$ then $\sum_{n=0}^{\infty}e^{-sn}a(n)=0$ and so
$a(n)=0$ for all $n$, that is, $a=0$. Thus $\mathcal{L}$ is
injective on $\mathcal{V}$.

The Laplace average (\ref{LaplaceAverage}) is related to the
Laplace transform of $E_{\xi}^A(n)$ by
\[L_{\xi}^A(T)=\frac{2}{T}\sum_{n=0}^{\infty}e^{-\frac{2n}{T}}E_{\xi}^A(n)=\frac{2}{T}f_{E_{\xi}^A}\left(\frac{2}{T}\right).\]

If $a(n)=1$ for all $n$, then
\[\frac{2}{T}f_{a}\left(\frac{2}{T}\right)=\frac{2}{T}\frac{1}{(1-e^{-2/T})}\approx\frac{2}{T}\frac{1}{2/T}=1,\]
for $T$ large enough. If $a(n)=(n+k)(n+k-1)\cdots(n+1)\approx
n^k,$ then
\[\frac{2}{T}f_{a}\left(\frac{2}{T}\right)=\frac{2}{T}\frac{k!}{(1-e^{-2/T})^{k+1}}\approx\frac{2}{T}\frac{k!}{(2/T)^{k+1}}=k!\left(\frac{T}{2}\right)^k,\]
for large $T$. Hence, if $E_\xi^A(T)$ grows like $T^k$ then the
same law holds for its average Laplace transform. We have a
restricted converse, that is, if $L^A_\xi(n)$ grows with a
positive power of $n$ then, by Lemma~\ref{lema.49}, its Ces\`aro
average is unbounded (with a rather similar behavior at large
times) and so is $E_\xi^A(n)$. These properties are repeated used
in the text.

One should be aware that there are special situations of unbounded positive sequences
$a(n)$ with bounded average Laplace transforms (so that
$\beta_e^+=\beta_d^+=0$); an explicit example is $a(n^2)=n$ and $a(n)=0$ for $n\notin \{k^2: k\in\N\}$. The same phenomenon is well known for
Ces\`aro averages and, by Lemma~\ref{lema.49}, such phenomena are  connected.


\end{document}